\documentclass[english,runningheads,11pt]{llncs-revised}
\usepackage{tabularx,booktabs,multirow,delarray,array}
\usepackage{graphicx,amssymb,amsmath,amssymb}
\usepackage{latexsym}
\usepackage[in]{fullpage}

\def\calP{\mathcal{P}}
\def\calR{\mathcal{R}}
\def\calF{\mathcal{F}}
\def\calC{\mathcal{C}}
\def\st{$s$-$t$}
\def\hatd{\hat{d}}

\newcommand{\vtd}{\mbox{$V\!D$}}
\newcommand{\htd}{\mbox{$H\!D$}}

\newtheorem{observation}{Observation}

\begin{document}

\title{An Optimal Algorithm for Minimum-Link Rectilinear Paths in Triangulated Rectilinear Domains\thanks{A preliminary version of
this paper will appear in the Proceedings of the 42nd International Colloquium on Automata, Languages, and Programming (ICALP 2015).
}
}

\author{Joseph S.B. Mitchell\inst{1}
\and
Valentin Polishchuk\inst{2}
\and
Mikko Sysikaski\inst{3}
\and
Haitao Wang\inst{4}
}

\institute{
Stony Brook University, Stony Brook,  NY 11794, USA\\
\email{jsbm@ams.stonybrook.edu}\\
\and
Link\"{o}ping University, Link\"{o}ping, Sweden\\
\email{valentin.polishchuk@liu.se}\\
\and
Google, Zurich, Switzerland\\
\email{mikko.sysikaski@gmail.com}\\
\and
Utah State University, Logan, UT 84322, USA\\
\email{haitao.wang@usu.edu}\\
}

\maketitle

\pagestyle{plain}
\pagenumbering{arabic}
\setcounter{page}{1}

\vspace{-0.2in}
\begin{abstract}
We consider the problem of finding minimum-link
rectilinear paths in rectilinear polygonal domains in the plane.
A path or a polygon is rectilinear if all its edges are axis-parallel.
Given a set $\calP$ of $h$ pairwise-disjoint rectilinear polygonal obstacles
with a total of $n$ vertices in the plane, a {\em minimum-link rectilinear
path} between two points is a rectilinear path that avoids all
obstacles with the minimum number of edges.
In this paper, we present a new algorithm for finding minimum-link rectilinear paths among $\calP$.
After the plane is triangulated, with respect to any source point $s$,  our algorithm
builds an $O(n)$-size data structure in $O(n+h\log h)$ time, such
that given any query point $t$, the number of edges of a minimum-link
rectilinear path from $s$ to $t$ can be computed in
$O(\log n)$ time and the actual path can be output in additional time
linear in the number of the edges of the path. The previously best algorithm
computes such a data structure in $O(n\log n)$ time.
\end{abstract}


\section{Introduction}
\label{sec:intro}

Paths with few turns have applications in a variety of areas,
including VLSI design, computer vision,
cartography, geographical information systems, robotics, computer
graphics, image processing, and solid modeling
\cite{ref:ArkinLo95,ref:DjidjevAn92,ref:GuibasAp93,ref:KolarovOn91,ref:McMasterAu87,ref:NatarajanOn91,ref:ReifMi87,ref:SuriMi87,ref:TamassiaOn86}.
Finding paths with few turns (or
minimum-link paths) has received much attention, e.g.,
\cite{ref:AdegeestMi94,ref:ArkinLo95,ref:DasGe91,ref:deBergOn91,ref:deRezendeRe89,ref:ElGindyHi85,ref:GhoshCo91,ref:ImaiEf86,ref:LingasOp95,ref:MaheshwariLi00,ref:MitchellMi14,ref:MitchellMi92,ref:ReifMi87,ref:SchuiererAn96,ref:SuriA86,ref:SuriMi87,ref:SuriOn90}.
In this paper, we enrich the literature by presenting a new algorithm for
finding the minimum-link rectilinear paths among rectilinear polygonal
domains in the plane.

Given a set $\calP$ of $h$ pairwise-disjoint polygonal obstacles with
a total of $n$ vertices in the plane,
the plane minus the interior of all obstacles is called the {\em free space}.
The {\em link distance} of a
path is defined to be the number of edges (also called {\em links})
in the path. A {\em minimum-link path} (or {\em min-link} path)
between two points $s$ and $t$ is a path from $s$ to $t$ in the free space
with the minimum link distance.
The {\em min-link path query} problem is to construct a data structure (called
{\em link distance map}) with respect to a given {\em source} point
$s$, such that for any query point $t$, a min-link path from $s$
to $t$ can be quickly computed. In the following, we say a
link distance map has the {\em standard query performance} if given any
$t$, the link distance of a min-link \st\ path can be computed in
$O(\log n)$ time and the actual path can be output in additional time
linear in the link distance of the path.

A polygon (or path) is {\em rectilinear} if all its edges are
axis-parallel.  $\calP$ is a {\em rectilinear
polygonal domain} if every obstacle of $\calP$ is rectilinear
(e.g., see Fig.~\ref{fig:rectilinear}).
The {\em rectilinear version} of the min-link path/query problem is to
find min-link {\em rectilinear} paths in rectilinear polygonal domains.
Rectilinear polygons are commonly used as
approximations to arbitrary simply polygons; and they arise naturally
in domains dominated by Cartesian coordinates, such as raster
graphics, VLSI design (as mentioned earlier), robotics, or architecture.

\begin{figure}[t]
\begin{minipage}[t]{\linewidth}
\begin{center}
\includegraphics[totalheight=1.1in]{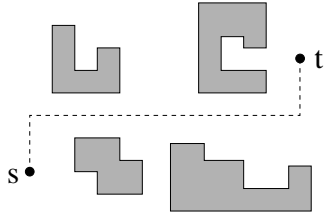}
\caption{\footnotesize Illustrating a rectilinear min-link path in
a rectilinear polygonal domain.}
\label{fig:rectilinear}
\end{center}
\end{minipage}
\vspace*{-0.15in}
\end{figure}

\subsection{Previous Work}

Linear-time algorithms have been given for finding min-link paths in simple polygons
\cite{ref:GhoshCo91,ref:HershbergerCo94,ref:SuriA86,ref:SuriMi87,ref:SuriOn90}.
The link distance map can also be built in
linear time \cite{ref:SuriA86,ref:SuriMi87,ref:SuriOn90} for simple
polygons, with the standard query performance.
For polygonal domains, the problem becomes much more
difficult.
Mitchell, Rote, and Woeginger \cite{ref:MitchellMi92}
gave an $O(n^2\alpha(n)\log^2n)$ time
algorithm for finding min-link paths, where $\alpha(n)$
is the inverse Ackermann function; a link distance map with slightly more
construction time is also given in \cite{ref:MitchellMi92}.
As shown by Mitchell, Polishchuk, and Sysikaski
\cite{ref:MitchellMi14}, finding
min-link paths in polygonal domains is 3SUM-hard.

The rectilinear min-link path problems have also been studied.
For simple rectilinear polygons, de Berg \cite{ref:deBergOn91}
presented an algorithm that can build an $O(n)$-size
link distance map in $O(n\log n)$ time and $O(n)$ space, with the standard query
performance. The construction time was later reduced
to $O(n)$ time by Lingas, Maheshwari, and Sack
\cite{ref:LingasOp95}, and by Schuierer \cite{ref:SchuiererAn96}.

For rectilinear polygonal domains, Imai and
Asano \cite{ref:ImaiEf86} presented an $O(n\log n)$  time and space
algorithm for finding min-link rectilinear paths.
Later, Das and Narasimhan \cite{ref:DasGe91} described an
improved algorithm of $O(n\log n)$ time and $O(n)$ space; Sato,
Sakanaka, and Ohtsuki \cite{ref:SatoA87} gave a similar algorithm with the same
performance.
Recently, Mitchell, Polishchuk, and Sysikaski \cite{ref:MitchellMi14}
presented a simpler algorithm of $O(n\log n)$ time and $O(n)$ space.
Link distance maps of $O(n)$-size can also be built in $O(n\log n)$
time and $O(n)$ space \cite{ref:DasGe91,ref:MitchellMi14}.
As shown in \cite{ref:DasGe91,ref:MaheshwariLi00}, the problem of finding min-link
rectilinear paths in rectilinear polygonal domains has a lower bound
of $\Omega(n+ h\log h)$ on the running time. Therefore, the algorithms in
\cite{ref:DasGe91,ref:MitchellMi14,ref:SatoA87} are optimal only if
$h=\Theta(n)$.

However, since the value $h$ can be substantially smaller than
$n$, it is desirable to have an
algorithm whose running time is bounded by $O(n+f(h))$, where $f(h)$
is a function of $h$.

In addition, the $C$-oriented version of the min-link paths problems have also been
considered \cite{ref:AdegeestMi94,ref:HershbergerCo94,ref:MitchellMi14},
where the edges of the paths are allowed to have $C$
different directions. Our rectilinear version is essentially a special
case of the $C$-oriented version with $C=2$.
	

\subsection{Our Results}
We consider the rectilinear min-link paths in a rectilinear domain $\calP$.
After the free space of $\calP$ is triangulated,
our algorithm builds a link distance map in
$O(n+h\log h)$ time and $O(n)$ space, with the standard query performance.
The triangulation can be done
in $O(n\log n)$ time or $O(n+h\log^{1+\epsilon}h)$
time for any $\epsilon>0$ \cite{ref:Bar-YehudaTr94}.
Hence, our result is an improvement over the previous $O(n\log n)$
time algorithms \cite{ref:DasGe91,ref:MitchellMi14,ref:SatoA87}, especially
 when $h$ is substantially smaller than $n$, e.g., if $h=O(n^{1-\delta})$ for any
$\delta>0$, our algorithm runs in $O(n)$ time.

\subsection{Algorithm Overview}
Our idea is to combine Das and
Narasimhan's algorithmic scheme
\cite{ref:DasGe91} and a {\em corridor structure} of polygonal
domains.  The corridor structure and its extensions have been used to solve
shortest path problems, e.g.,
\cite{ref:ChenTw14,ref:ChenA11ESA,ref:KapoorAn97,ref:MitchellSe95};
however, to the best of our knowledge, it was not used to tackle
min-link path problems. The corridor structure partitions the free
space of $\calP$ into $O(h)$ corridors and $O(h)$ ``junction''
rectangles that connect all corridors.

The algorithm in \cite{ref:DasGe91} (which we call the DN algorithm) sweeps the free space, from the source point $s$, to build the map.
The sweep is controlled in a global way so that the time is bounded by
$O(n\log n)$. This global sweeping on the entire free space restricts the DN algorithm from being implemented in $O(n+h\log h)$ time because each operation takes $O(\log n)$ time and there are $O(n)$ operations. Using the corridor structure, our algorithm avoids the global sweeping on the entire free space.
When the sweep is in junction rectangles, we control the sweep in a global way as
in the DN algorithm. However, when the sweep enters a corridor,
we process the corridor independently and ``locally'' without considering the
space outside the corridor. Since a
corridor is a simple polygon, we are able to design a
faster algorithm for processing the sweep in it.
When we finish processing a corridor, we arrive at a junction rectangle.
Next, we pick an unprocessed junction
rectangle that currently has the smallest link distance to $s$ to
``resume'' the sweep. This is somewhat similar to Dijkstra's shortest path algorithm.
In this way, there are only $O(h)$ operations that need to be performed in logarithmic time each.

To achieve the $O(n+h\log h)$ time bound, we have to implement the algorithm in a very careful manner. For example, we need an efficient
algorithm to compute link distance maps in corridors (i.e.,
simple rectilinear polygons).
Although efficient algorithms are available for computing the link
distance maps in simple polygons, e.g.,
\cite{ref:deBergOn91,ref:HershbergerCo94,ref:LingasOp95,ref:SchuiererAn96},
they are not suitable for our main algorithmic scheme, which requires
an algorithm with some special properties. Specifically, let
$\calC$ be a corridor. Suppose there are $h_1$
pairwise-disjoint segments on a vertical edge $d_1$ (called a
``door'') of $\calC$ and these segments are considered as ``light sources'', stored in a balanced binary search tree.  We
want to build a link distance map in $\calC$, and obtain a balanced
binary search tree storing the light sources (let $h_2$ be their number)
going out of $\calC$ through another vertical edge $d_2$
(another ``door'') of $\calC$. Our goal is to design
a ``corridor-processing'' algorithm for the above problem with running time
$O(m+(h_1-h_2+1)\log h_1)$, where $m$ is the number of vertices of $\calC$. We derive such an
algorithm, which might be of independent interest.


The rest of the paper is organized as follows.
In Section \ref{sec:pre}, we define notation and
review the DN algorithm \cite{ref:DasGe91}.
In Section \ref{sec:mainscheme}, we introduce the corridor structure
and present the main scheme of our algorithm while leaving our
algorithms for processing corridors in Section \ref{sec:algocorridor}.
Section \ref{sec:conclude} concludes the paper.

\section{Preliminaries}
\label{sec:pre}

For simplicity of discussion,
let $\calR$ be a large rectangle that contains all obstacles of
$\calP$ and let $\calF$ denote the free space of $\calP$ in $\calR$. Note that our algorithm can handle the case where $\calR$ is an arbitrary rectilinear polygon, but for simplicity of discussion, we consider the case where $\calR$ is a rectangle.
We assume $\calF$ has been triangulated. Let $s$ be a given source
point in $\calF$.  For ease of exposition, we make a general position
assumption that no three vertices of $\calP\cup\{s\}$ have
the same $x$- or $y$-coordinate; the assumption can be lifted without
affecting the performance of the algorithm asymptotically although
more tedious discussions would be needed.
In the following, paths always refer to
rectilinear paths in $\calF$. For any point $t$ in
$\calF$, a path from $s$ to $t$ is also referred to as an $s$-$t$
path. An edge of a path is also called a {\em link} of the path.

Consider any point $t\in \calF$.
An \st\ path $\pi$ is called a {\em horizontal-start-vertical-end}
path (or {\em h-v-path} for short)
if the first link of $\pi$ (i.e., the edge incident to $s$) is
horizontal and the last link of  (i.e., the edge incident to $t$) is
vertical. The {\em h-h-paths}, {\em v-h-paths}, and {\em v-v-paths} are
defined analogously.
To make it consistent, if $\pi$ is an h-h-path of $k$ links,
we also consider it to be an h-v-path of $k+1$ links (i.e., we
enforce an additional edge of zero length at the end of the path), and
similarly, it is also considered to be an v-h-path of $k+1$ links and
a v-v-path of $k+2$ links. A {\em min-link h-v-path} from $s$ to $t$ is an h-v-path from $s$ to $t$ with
the minimum number of links. The {\em min-link h-h-paths}, {\em min-link v-h-paths}, and {\em min-link v-v-paths} are
defined analogously.
To find a min-link \st\ path, our algorithm will
find the following four \st\ paths: a min-link h-v-path,   a min-link
h-h-path,  a min-link v-v-path, and a min-link v-h-path. Clearly,
among the above four paths, the one with the minimum number of links
is a min-link \st\ path.

To answer the min-link query, our algorithm will compute four link
distance maps of $O(n)$ size each: an {\em  h-h-map}, an {\em h-v-map}, a {\em v-h-map},
and a {\em v-v-map}, defined as follows. The h-h-map is a
decomposition of $\calF$ into regions such that for any region $R$,
the link distances of the min-link h-h-paths from $s$ to all points in
$R$ are the same. The other three maps are defined analogously.
In addition, we build linear-size point location data structures
\cite{ref:EdelsbrunnerOp86,ref:KirkpatrickOp83} on these maps in linear
time.  With the above four maps, for any query point $t$, we determine the region
containing $t$ in each map and the one
with the smallest link distance gives our sought min-link \st\ path
distance.

The {\em vertical visibility decomposition} of $\calF$, denoted by
$\vtd(\calF)$,
is obtained by extending each vertical edge of the obstacles in
$\calP$ until it hits either another obstacle or
the boundary of $\calR$ (e.g., see Fig.~\ref{fig:vtd}).
We call the above edge extensions the {\em diagonals}.
We consider the source $s$ as a special obstacle and extend a vertical diagonal
through $s$.
Since $\calF$ has been triangulated,  $\vtd(\calF)$ can be obtained in
$O(n)$ time \cite{ref:Bar-YehudaTr94,ref:ChazelleTr91,ref:ChazelleTr84,ref:FournierTr84}.
In $\vtd(\calF)$, $\calF$ is decomposed into rectangles,
also called {\em cells}. Due to our general position
assumption, each vertical side of a cell may contain at most two diagonals.
The horizontal visibility decomposition of $\calF$, denoted by $\htd(\calF)$, is
defined similarly by extending the horizontal edges of $\calP$.

Our v-v-map is on $\vtd(\calF)$, i.e., for each cell $C$ of
$\vtd(\calF)$, the
link distances of the min-link v-v-paths from $s$ to all points in $C$
that are not on diagonals
are the same and we denote this distance by $dis_{vv}(C)$, and for
each diagonal $d$ of $\vtd(\calF)$, the
link distances of the min-link v-v-paths from $s$ to all points on $d$
are the same and we denote this distance by $dis_{vv}(d)$. In fact, it holds that $dis_{vv}(d)=\min\{dis_{vv}(C_l), dis_{vv}(C_r)\}$,
where $C_l$ and $C_r$ are the two cells on the left and right of $d$,
respectively.
The goal of our algorithm for computing the v-v-map is to compute
$dis_{vv}(C)$ for each cell $C$ and $dis(d)$ for each
diagonal $d$ of $\vtd(\calF)$.  In
addition, we also need to maintain some path information to retrieve
an actual path for each query.

Similarly, our h-v-map is also on $\vtd(\calF)$, but the h-h-map and the v-h-map
are both on $\htd(\calF)$.

\begin{figure}[t]
\begin{minipage}[t]{\linewidth}
\begin{center}
\includegraphics[totalheight=1.8in]{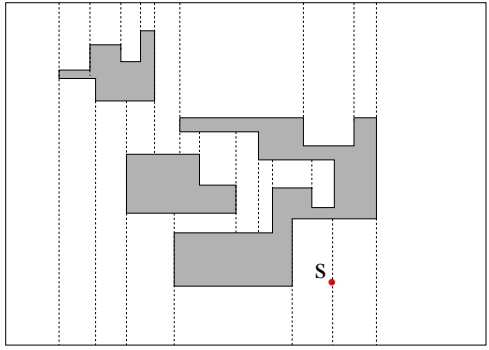}
\caption{\footnotesize Illustrating the vertical visibility
decomposition of $\calF$. }
\label{fig:vtd}
\end{center}
\end{minipage}
\vspace{-0.15in}
\end{figure}


In Section \ref{sec:DN}, we review and discuss the DN algorithm \cite{ref:DasGe91} in a way that will help us to introduce our algorithm later. The DN algorithm only labels the diagonals of $\vtd(\calF)$.
Das and Narasimhan \cite{ref:DasGe91} claimed that queries can be
answered by using only the distance values on the diagonals. However,
this is not clear to us. We, instead, also label the cells
as discussed above. We will show in Section \ref{sec:labelcell} that the DN algorithm can be easily
adapted to computing the link distances for the cells too. We will also discuss in Section \ref{sec:labelcell} how to maintain path information to retrieve an actual path for each query.

\subsection{The DN Algorithm}
\label{sec:DN}

The DN algorithm also computes the four maps discussed above. We consider the
v-v-map first. The algorithm works on the vertical visibility
decomposition $\vtd(\calF)$ and the goal is to compute $dis_{vv}(d)$ for each
diagonal $d$.
To simplify the notation, unless otherwise
stated, we use $dis(\cdot)$ to refer to $dis_{vv}(\cdot)$.
Initially, all diagonals have distance value
$\infty$ except $dis(d_s)=1$, where $d_s$ is the diagonal through $s$.
Note that if a diagonal $d$ is on a side $e$ of a cell, then
whenever $dis(d)$ is updated, $dis(e)$ is automatically set to $dis(d)$.

The DN algorithm has many phases. In the $i$-th phase for $i\geq 0$, the
algorithm determines the set $V_i$ of diagonals $d$ whose distances
$dis(d)$ are equal to $2i+1$, and these diagonals are then ``labeled''
with distance $2i+1$ (e.g., see Fig.~\ref{fig:DN}). Initially, $i=0$, and $V_0$ consists of
the only diagonal through $s$. As discussed in \cite{ref:DasGe91},
if we put light sources on the diagonals in $V_{i-1}$,
then $V_i$ consists of all new diagonals that will get illuminated with light emanating horizontally from the light sources.

Consider a general $i$-th phase for $i\geq 1$.
We assume $V_{i-1}$ has been determined. There are
two procedures: {\em right-sweep} and {\em left-sweep}. In the
right-sweep (resp., left-sweep), we illuminate the diagonals in the direction
from left to right (resp., from right to left).
The right-sweep procedure starts from the
{\em locally-rightmost} diagonals of $V_{i-1}$, defined as follows.
Consider any diagonal $d$ in $V_{i-1}$. Let $C$ be the cell of $\vtd(\calF)$ on the
right of $d$, i.e., $d$ is on the left side of $C$.
Let $e_r$ be the right side of $C$. If $dis(e_r)\neq 2i-1$\footnote{
Alternatively, one may replace $dis(e_r)\neq 2i-1$ by $dis(e_r)>2i-1$
for the definition, but we choose to use $dis(e_r)\neq 2i-1$ only for
making our discussion later easier and this will not affect the
algorithm running time.}, then $d$ is a {\em locally-rightmost} diagonal of $V_{i-1}$.
Similarly, the left-sweep starts from the {\em
locally-leftmost} diagonals of $V_{i-1}$. Both
locally-leftmost and locally-rightmost diagonals are referred to as
{\em locally-outmost} diagonals. Below, we first discuss the
right-sweep.

\begin{figure}[t]
\begin{minipage}[t]{0.60\linewidth}
\begin{center}
\includegraphics[totalheight=2.0in]{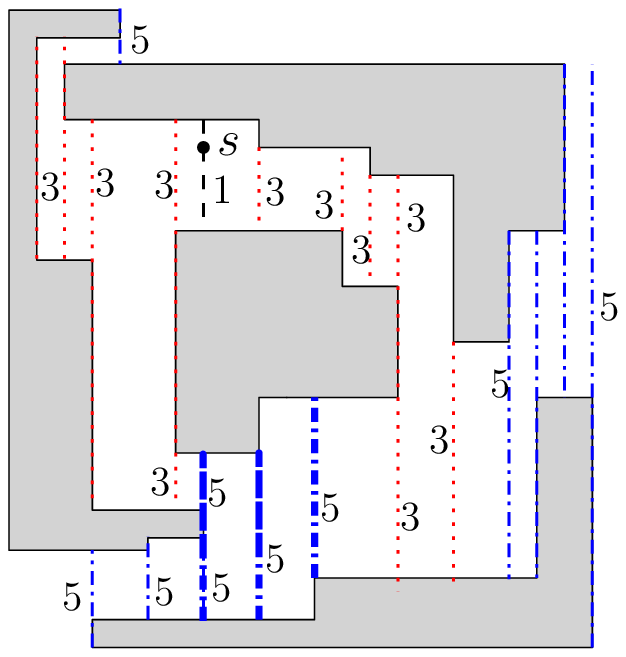}
\caption{\footnotesize Illustrating the DN algorithm: all diagonals
have been labeled. Note that the three thick dash-dotted diagonals
(labeled $5$) are swept twice in the second phase.}
\label{fig:DN}
\end{center}
\end{minipage}
\hspace*{0.05in}
\begin{minipage}[t]{0.38\linewidth}
\begin{center}
\includegraphics[totalheight=1.5in]{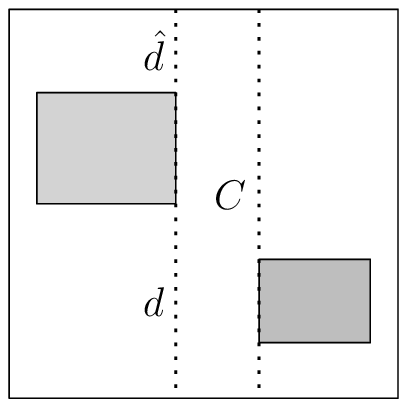}
\caption{\footnotesize Illustrating a cell $C$ where $d$ is on the
left side of $C$.}
\label{fig:cell}
\end{center}
\end{minipage}
\vspace{-0.15in}
\end{figure}

For each
locally-rightmost diagonal $d$, we put a rightward ``light beam'' on
$d$, and let $B(d)$ denote the beam;
we also call $d$ the {\em generator} of the beam.
Formally, we may define the beam as the segment $d$ associated with the rightward
direction. Initially we insert all locally-rightmost
diagonals into a min-heap $H_R$ whose ``keys'' are the $x$-coordinates of the
diagonals (i.e., the leftmost diagonal is at the root).
By using $H_R$, the diagonals involved in the right-sweep will be processed from left to
right. As discussed in \cite{ref:DasGe91}, this can control the
right-sweep in a global manner and thus avoid the same
diagonal being processed many times.
Although initially each diagonal of $H_R$ has only one
beam, later we will insert more diagonals into $H_R$ and a
diagonal $d$ may have more than one beam,
in which case $B(d)$ denotes the set of beams on $d$.

As long as $H_R$ is not empty, we repeatedly do the following.

We obtain the leftmost diagonal $d$ of $H_R$ and remove it from $H_R$.
Let $C$ be the cell on the right of $d$ (e.g., see Fig.~\ref{fig:cell}).
We {\em process} $d$ in the following way. Intuitively we want to
propagate the beams of $B(d)$ to other diagonals in $C$.
Let $e_l$ and $e_r$ denote the left and right sides of $C$, respectively.
Note that $d$ is on $e_l$.  Recall that each cell side has at
most two diagonals.

If $e_l$ has another diagonal $\hatd$ (e.g., see Fig.~\ref{fig:cell}) and $\hatd$ has
not been labeled (i.e., $dis(\hatd)=\infty$), 
then we set $dis(\hatd)=2i+1$. The beam set
$B(\hatd)$ of $\hatd$ is set to $\emptyset$
since there is no beam from $B(d)$ illuminating $\hatd$.
Further, although $B(\hatd)=\emptyset$, we associate the {\em
leftward direction} with it, because
if later the left-sweep procedure does not illuminate $\hatd$, $\hatd$
will be a locally-leftmost diagonal of $V_i$ and
generate a leftward beam in the next phase.
We ``temporarily'' mark $\hatd$ as a locally-leftmost
diagonal.
We say ``temporarily'' because $\hatd$ may not be
locally-leftmost any more after the left-sweep procedure, as discussed later.

If $\hatd$ has been labeled, then $dis(\hatd)$ must be $2i+1$. To see this, because we are currently at the $i$-th phase, any diagonal that has been labeled must have a distance value at most $2i+1$. But if $dis(\hatd)<2i+1$, then $\hatd$ has already been processed, and since $\hatd$ and $d$ are on the same cell side, we must have already obtained $dis(d)=dis(\hatd)<2i+1$, contradicting with $dis(d)=2i+1$. Hence, $dis(\hatd)=2i+1$. In this case, we do nothing on $\hatd$.

Next, we consider the diagonals on
the right side $e_r$ of $C$. Depending on the values
of $dis(e_r)$, there are several cases. Since we are at the $i$-th
phase, either $dis(e_r)=\infty$ or $dis(e_r)\leq 2i+1$.

\begin{enumerate}
\item
If $dis(e_r)<2i+1$, then we do nothing on $e_r$ because all diagonals on $e_r$ have already
been processed.

\item
If $dis(e_r)=\infty$,
then we set $dis(e_r)=2i+1$. If $e_r$ does not have any diagonals,
we are done. Otherwise, for each diagonal
$d'$ on $e_r$, we determine the portions of beams of $B(d)$ that can
illuminate $d'$, which are the rightward projections of $B(d)$ on
$d'$ (beams of $B(d)$ may be ``narrowed'' or ``split'', e.g., see Fig.~\ref{fig:beams}). We use
$B(d)\cap d'$ to denote the above portions of $B(d)$. It is
possible that $B(d)\cap d'$ is empty.

If $B(d)\cap d'=\emptyset$, then we temporarily mark $d'$ as a locally-rightmost diagonal
and set $B(d')=\emptyset$ with the {\em rightward
direction}; otherwise, we set $B(d')=B(d)\cap d'$ and insert $d'$ to
$H_R$ (later we will propagate the beams of $B(d')$ further to the
right).

\begin{figure}[t]
\begin{minipage}[t]{\linewidth}
\begin{center}
\includegraphics[totalheight=1.3in]{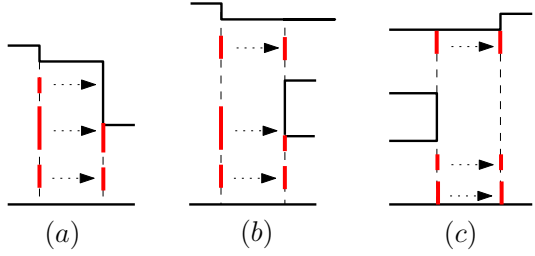}
\caption{\footnotesize Illustrating some beam operations in a right-sweep procedure: In (a) and (b), beams are {\em split}, and some beams are ``narrowed'' and some beams ``terminate'' at obstacle edges; in (c), beams are {\em merged}.}
\label{fig:beams}
\end{center}
\end{minipage}
\vspace{-0.15in}
\end{figure}

\item
If $dis(e_r)=2i+1$, this case happens because $e_r$
was illuminated by beams from another diagonal $\hatd$ on $e_l$.
Hence, each diagonal $d'$ on
$e_r$ may already have a non-empty $B(d')$. But $d'$ may receive
more beams from $B(d)$. We first determine $B(d)\cap d'$ (as defined
above) and then
do a ``merge'' operation by merging $B(d)\cap d'$ with $B(d')$ (e.g., see Fig.~\ref{fig:beams}).
Finally, we set $B(d')$ to the above merged set of beams (with the
rightward direction). Due to the merge operation,
$B(d')$ may contain more than one beam.

If $B(d')$ was empty before the merge and becomes non-empty after the
merge, then we insert $d'$ into $H_R$. If $B(d')$ is still empty after
the merge, then we temporarily mark $d'$ as a locally-rightmost diagonal.
If $B(d')$ was non-empty before the merge, then $d'$ is already in
$H_R$, so we do not need to insert it into $H_R$ again.
\end{enumerate}

The above finishes the processing of $d$.  The right-sweep is done
once $H_R$ becomes empty. For the implementation, we
need to maintain the beams of $B(d)$ on each diagonal $d$. To this end,
as in \cite{ref:DasGe91}, we can use
a balanced binary search tree (e.g., a 2-3-4 tree \cite{ref:CLRS09}) such that we can perform
the ``merge'', ``split'', ``insertion'',
``deletion'', and ``search'' operations each in logarithmic time.

The left-sweep procedure is similar. However, there is one
subtle thing. If the sweep encounters a diagonal $d$ that has been
labeled by the right-sweep, then this is ignored and
we keep sweeping as if $d$ were not labeled.
As discussed in
\cite{ref:DasGe91}, the reason for this is that the left-sweep may reach more
cells than the right-sweep (e.g., see Fig.~\ref{fig:DN}). In this way, each diagonal
can be processed at most twice in a phase. But no diagonal can be
processed in more than one phase.
Also, suppose a diagonal $d$ was temporarily marked as a locally-outmost diagonal during the right-sweep; if $d$ is illuminated again in the left-sweep but $d$ is not marked locally-outmost in the left-sweep, then we clear the previous mark on $d$ (i.e., $d$ is not considered locally-outmost  any more).
After the left-sweep, the
remaining locally-outmost diagonals are exactly those that will be
used in the next phase.

The above describes the $i$-th phase of the algorithm. The algorithm
is done after all diagonals are labeled.
Due to that heaps are used to control the sweep procedures and balanced binary
search tree are used to support beam operations, the total time of the
algorithm is $O(n\log n)$ because each diagonal can be processed at most
twice in the entire algorithm (once by a left-sweep procedure and
once by a right-sweep procedure). Clearly, the space is $O(n)$.

For computing the h-v-map, the algorithm is similar except that $V_0$ now consists of all
diagonals that intersect $d'_s$, where $d'_s$ is the horizontal line segment
extended from $s$ until it hits the obstacles, and in each $i$-th
phase for $i\geq 0$, the value $dis_{hv}(d)$ is set to $2(i+1)$ for
any diagonal $d$ in $V_i$. The algorithms for computing
h-h-map and v-h-map in $\htd(\calF)$ are symmetric.

\subsection{Labeling Cells and Retrieving Paths}
\label{sec:labelcell}

We first show how to slightly modify the DN algorithm to label cells.
We only discuss the v-v-map since the other maps are similar.

Recall that for each cell $C$ of $\vtd(\calF)$, all the points
in $C$ that are not on the diagonals have the same distance values.
Hence, to compute $dis_{vv}(C)$, it is sufficient to know the distance
value for any arbitrary point in $C$ that is not on any diagonal. To this end,
for each cell $C$, we add a
vertical segment in $C$ with its upper endpoint on the upper side of
$C$ and its lower endpoint on the lower side of $C$ such that the
segment is not overlapped with the left side or the right side of $C$;
hence, no point of the segment is on any diagonal, and
we call the segment a ``fake diagonal''. Let $\vtd'(\calF)$ denote
$\vtd(\calF)$ with all fake diagonals.
We run the DN algorithm on $\vtd'(\calF)$ and treat all fake diagonals
as ``ordinary diagonals'' to label all diagonals and fake diagonals of
$\vtd'(\calF)$. Finally, we label each cell of
$\vtd(\calF)$ with the same distance value as its fake diagonal.
Since the total number of fake diagonals are $O(n)$, the running time is still
$O(n\log n)$ and the space is $O(n)$.

To obtain an actual min-link v-v-path from $s$ to $t$, we need to
maintain additional information on our v-v-map. No details
on this are given in \cite{ref:DasGe91}.
We briefly discuss it below for the completeness
of this paper. Essentially, when we
label the cell sides (and the fake diagonals), we need to record how
we reach there. Specifically, suppose we label a diagonal $d$ on the
right side of a cell in a right-sweep procedure due
to a beam from a diagonal on the left side of the cell; then we record any such
beam at $d$ (it is sufficient to record any
point on $d$ in the beam) along with its
generator, and in the case where $B(d)$ is empty,
$d$ is a locally-outmost diagonal and the path should make a turn
there. Further, for each locally-outmost diagonal $d$, it is reached by a beam that illuminates the
cell side containing $d$ and we record that beam for $d$ so that we know the turn on $d$ is for that beam.
With this path information, for any query point $t$,
we can easily retrieve an actual
min-link v-v-path from $t$ back to $s$ in time we claimed before. We omit the
details.

\section{The Main Scheme of Our Algorithm}
\label{sec:mainscheme}


In this section, we focus on the main scheme of our algorithm while leaving
the algorithms for processing corridors in
Section \ref{sec:algocorridor}. We first introduce the corridor
structure.

\subsection{The Corridor Structure}
\label{sec:corridor}

The corridor structure in rectilinear domains is similar to that
in general polygonal domains \cite{ref:KapoorAn97}.
Let $G_{vtd}$ be the dual graph of the vertical visibility
decomposition $\vtd(\calF)$ (see
Fig.~\ref{fig:graphG}), i.e., each node of
$G_{vtd}$ corresponds to a cell of  $\vtd(\calF)$,
and each edge of $G_{vtd}$ connects two
nodes corresponding to two cells sharing the same diagonal.
Based on $G_{vtd}$, we obtain a {\em corridor graph} $G$ as follows (see
Fig.~\ref{fig:reduceG}). First, we
remove every degree-one node from $G_{vtd}$ along with its incident edge;
repeat this process until no degree-one node exists. Second,
remove every degree-two node from $G_{vtd}$
and replace its two incident edges by a
single edge; repeat this process until no degree-two node exists. The
remaining graph is $G$. The cells in $\vtd(\calF)$ corresponding to the nodes
in $G$ are called {\em junction cells} (see Fig.~\ref{fig:reduceG}).
We consider the diagonal through $s$ as a degenerate junction cell.
Similar to the corridor structure in the
general polygonal domains \cite{ref:KapoorAn97},
the graph $G$ has $O(h)$ nodes and $O(h)$ edges.
The removal of all junction cells from $\vtd(\calF)$ results in $O(h)$
{\em corridors}, each of which corresponds to an edge of $G$.

\begin{figure}[t]
\begin{minipage}[t]{0.49\linewidth}
\begin{center}
\includegraphics[totalheight=1.6in]{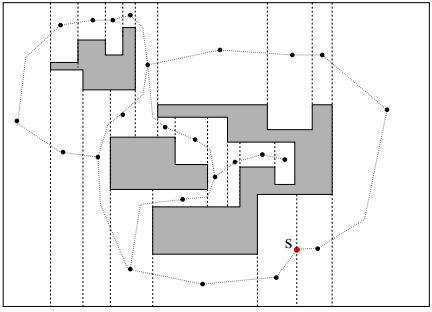}
\caption{\footnotesize Illustrating the vertical visibility
decomposition (the
dashed segments are diagonals) and its dual graph $G_{vtd}$. }
\label{fig:graphG}
\end{center}
\end{minipage}
\hspace{0.02in}
\begin{minipage}[t]{0.49\linewidth}
\begin{center}
\includegraphics[totalheight=1.6in]{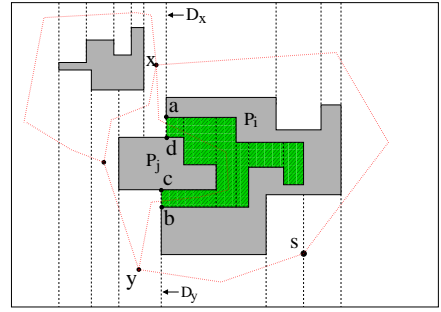}
\caption{\footnotesize Illustrating the graph $G$,
and the corridor (shaded by slashes)
bounded by $P_i$ and $P_j$.
 }
\label{fig:reduceG}
\end{center}
\end{minipage}
\end{figure}

The boundary of any corridor $\calC$ consists of four parts (see
Fig.~\ref{fig:reduceG}): (1) The boundary portion of an obstacle $P_i$, from
a point $a$ to a point $b$; (2) a diagonal $\overline{bc}$; (3) the
boundary portion of an obstacle $P_j$ from $c$ to a point $d$; (4) a
diagonal $\overline{da}$.  The two diagonals $\overline{bc}$
and $\overline{ad}$ are called the {\em doors} of the corridor $\calC$.
The corridor $\calC$ is a simple rectilinear polygon.

\subsection{The Main Idea}
\label{sec:main}

We focus on computing the v-v-map on $\vtd(\calF)$, and
the other three maps can be computed similarly.
Our goal is to label $d$ for each diagonal $d$ of $\vtd(\calF)$, i.e., compute the
distance value $dis(d)=dis_{vv}(d)$.
As discussed in Section \ref{sec:labelcell}, the same algorithm
can be used to label cells as well. As in the DN algorithm
each diagonal $d$ will maintain a beam set $B(d)$.


Here is some intuition that motivates us to improve the DN algorithm.
When we are running the DN algorithm, a sweep
procedure will enter each corridor though one of its two doors, and the
procedure will either sweep the entire corridor and leave the corridor
through the other door, or terminate inside the corridor (in which case
the sweep ``hits'' another sweep that entered the corridor through the
other door and both sweeps terminate after ``collision'').  This means that if we can
determine the beams and the distance values at the doors of a corridor,
then we can process the corridor independently in a more efficient way since
the corridor is a simple rectilinear polygon.

Our algorithm follows the similar scheme as the DN algorithm. The
sweep in the junction cells is still processed and
controlled in a global manner in each phase.
However, whenever the sweep enters a corridor through one of its doors, the corridor will
be processed independently by using our more efficient {\em
corridor-processing} algorithm given in Section \ref{sec:algocorridor}
(one may view that the sweep procedure jumps from one junction cell to
another through the corridor connecting them).
We should point out that since
in the DN algorithm a diagonal may be processed twice in the two
sweep procedures in the same phase, an entire corridor may be processed twice in
the same phase (that case happens {\em only if} the beams on a door
can illuminate the other door directly, and vice versa).

The running time of our algorithm is $O(n+h\log h)$. More specifically, since
there are $O(h)$ junction cells, the time spent on processing the
diagonals in all
junction cells is $O(h\log h)$, and the processing on all corridors
takes $O(n)$ time because the number of vertices of all
corridors is $O(n)$, in addition to another $O(h\log h)$ time spent on
maintaining the beams on all diagonals.

\subsection{The Algorithm}

To help the reader understand the algorithm, we first describe the first few
steps of the algorithm and then delve into the details of the general steps.

Initially, we set $dis(d)$ to $\infty$ and $B(d)=\emptyset$ for each diagonal $d$ except
that $dis(d_s)=1$ and $B(d_s)=\{d_s\}$, where $d_s$ is the diagonal through $s$.
In our discussion of the corridor structure, $d_s$ is
considered as a degenerate junction cell. Here, for ease of
discussion, we assume $d_s$ is in the interior of a corridor $\calC(s)$ as shown
in Fig.~\ref{fig:reduceG}. We first process $\calC(s)$, i.e., label
all diagonals in $\calC(s)$.  For the purpose of describing our algorithm,
one may assume we still use the DN algorithm to process $\calC(s)$,
and later we will replace the DN algorithm by our
new and more efficient corridor-processing algorithm in Section \ref{sec:algocorridor}.

Denote by $D_x$ and $D_y$ the two doors of $\calC(s)$ (e.g., see
Fig.~\ref{fig:reduceG}). Suppose after the above processing of
$\calC(s)$, $dis(D_x)=2i+1$ and $dis(D_y)=2j+1$ for some integers $i$
and $j$.  Since $dis(D_x)$ is obtained ``locally'' in $\calC(x)$,
i.e., $dis(D_x)$ is the link distance of the {\em local} min-link v-v-path from
$s$ to $D_x$ in $\calC(s)$, it may not be ``set correctly'', i.e., it may not be
the link distance of the {\em global} min-link v-v-path from  $s$ to $D_x$ in
$\calF$. The value $dis(D_y)$ has the same issue. However, we show
below the at least one of $dis(D_x)$ and $dis(D_y)$ must have been set
correctly.

Without loss of generality, assume $j\leq i$.

\begin{observation}\label{obser:10}
If $j\leq i$, then the value $dis(D_y)$ has been set correctly.
\end{observation}
\begin{proof}
Let $\pi$ be any global min-link v-v-path from $s$ to $D_y$ in $\calF$. For any subpath $\pi'$ of $\pi$,
we use $dis(\pi')$ to denote the link distance of $\pi'$.
If $\pi$ is in
$\calC(s)$, then our observation follows. Otherwise, $\pi$ must cross one of
the doors of $\calC(s)$. Let $p$ be the first
point on a door of the corridor if we go from $s$ to $D_y$ along
$\pi$, and let $\pi(s,p)$ denote the sub-path of $\pi$ from $s$ to $p$. If $p$ is on $D_y$, then $\pi(s,p)$ is a min-link v-v-path from $s$
to $D_y$. Since $\pi(s,p)\in \calC(s)$ and $dis(D_d)$ is the link distance of the local
min-link v-v-path from $s$ to $D_y$ in $\calC(s)$, we have $dis(\pi)\geq dis(\pi(s,p))\geq dis(D_y)$, and thus our observation follows. If $p$
is on $D_x$, then $dis(\pi(s,p))=2i+1\geq 2j+1$. Therefore, $dis(\pi)\geq
dis(\pi(s,p))\geq 2j+1=dis(D_y)$ and our observation also follows.\qed
\end{proof}

Another way to see the observation is that if we were running the DN algorithm,
after the $(j-1)$-th phase, neither $D_x$ nor $D_y$ is labeled, and
no diagonal outside the corridor is labeled either. In the $j$-th
phase, $D_y$ will be labeled and thus $dis(D_y)=2j+1$ is set
correctly.

After the processing of $\calC(s)$, the beams on $D_y$, i.e.,
$B(D_y)$, have also been obtained.
The next step is to process the diagonal $D_y$ by propagating the
beams of $D_y$ outside the corridor.


Next, we describe the details of the general steps of our algorithm.

We use a min-heap $H$ to maintain the diagonals in all junction cells where the ``keys'' are the distance values the diagonals currently have (and these values may not be set correctly), with the {\em smallest} key at the root of $H$.  Since there are $O(h)$ junction cells, the size of $H$ is always $O(h)$.
Each diagonal $d$ in $H$ is also associated with its beam set
$B(d)$ (along with its direction).
As in the DN algorithm, it is possible that $B(d)$ is empty,
in which case $d$ might be a locally-outmost
diagonal, but it is also associated with a
direction for generating a beam towards that direction  in the next phase.

If some diagonals of $H$ have the same
keys, we break the ties by applying the following rules. Consider
two diagonals $d_1$ and $d_2$ in $H$ with $dis(d_1)=dis(d_2)$. If
$B(d_1)$ and $B(d_2)$ are both empty or both non-empty, then we break ties
arbitrarily. Otherwise, without loss of generality, assume $B(d_1)$ is
not empty but $B(d_2)$ is empty. Then, we consider the key of $d_1$ to be {\em smaller}
than that of $d_2$ in $H$.
The reason is as follows. Since $B(d_1)\neq \emptyset$ and $B(d_2)=\emptyset$,
the current sweep procedure should be over {\em before}
processing $d_2$ while the sweep
should continue {\em after} processing $d_1$, and thus, we
should process $d_1$ before $d_2$.  Therefore,
our way of resolving ties in $H$ is crucial and consistent with the DN algorithm.
In the following, to differentiate from keys of other heaps, for any diagonal $d$, even if $d$ is not in $H$, we consider $dis(d)$ along with $B(d)$ as the {\em global-key} of $d$, and whenever we compare the global-keys of diagonals, we always follow the above rules.

Consider any corridor $\calC$ with two doors $d$ and $d'$. Suppose
the beams of $B(d)$ are going inside $\calC$, and we want to
{\em process} $\calC$ (i.e., compute the v-v-map in $\calC$) using the beams of $B(d)$.
We say the above way of processing $\calC$ is in the {\em direction} from
$d$ to $d'$. As will be seen later,
a corridor may be processed twice: once
from $d$ to $d'$ and the other from $d'$ to $d$. Due
to the special geometric structure of the corridor, we have the
following observation that will be useful later.

\begin{observation}\label{obser:20}
Suppose $d$ and $d'$ are the two doors of a corridor $\calC$, and the
direction of the beams of $B(d)$ is towards the inside of $\calC$. Then after $\calC$ is
processed by using $B(d)$, the beam set of $d'$ is not empty.
\end{observation}
\begin{proof}
Let $\vtd(\calC)$ denote the vertical visibility decomposition of the
corridor $\calC$.
Consider the cell $C$ of $\vtd(\calC)$ that contains $d'$.
Without loss of generality, assume $d'$ is on the right side of
$C$. Denote by $e_r$ the right side of $C$.

We claim that $d'$ is the entire right side of $C$, i.e., $d'$ is $e_r$.
Indeed, according to our corridor structure, $d'$ is an extension of a
vertical obstacle edge $e$ and one endpoint $p$ of $e$ is also an endpoint of
$d'$ and the other endpoint of $e$ is outside the corridor $\calC$. This
means $p$ is also an endpoint of $e_r$ and
the other endpoint $q$ of $d'$ than $p$ is not on $e$ but on
another obstacle edge $e'$. Due to our general position assumption
that no two vertical edges are collinear, $q$ must be in the interior of
$e'$, implying that $q$ is also an endpoint of $e_r$. Hence, we obtain $d'=e_r=\overline{pq}$.

Note that we obtain the beam set of $d'$ from the rightward beams of
the diagonals on the left side of $C$. Now that $d'$ is
the entire right side of $C$, $d'$ will receive all beams of any
diagonal on the left side of $C$. Hence, the beam set of $d'$ cannot be empty.\qed
\end{proof}

To show the correctness of our algorithm, we will argue that
our algorithm is consistent with the DN algorithm.
We will show that after the algorithm finishes, each diagonal
in any junction cell is {\em correctly labeled}, i.e.,
both its distance value and its beam set are the same as those in the
DN algorithm.  Let $d^*$ be the diagonal in the root of $H$.
Our algorithm will maintain the following three invariants.

\begin{enumerate}
\item
The diagonal $d^*$ is correctly labeled. Further,
for any other diagonal $d$ in a junction cell, if the global-key of $d$ is no larger than that of $d^*$, then $d$ has been  correctly labeled.

%

\item
For any diagonal $d$ in a junction cell, if $dis(d)\neq \infty$ and
the global-key of $d$ is larger than that of $d^*$, then $d$ is in $H$.

\item
For any corridor $\calC$ with two doors $d$ and $d'$, if $\calC$ is
processed in the direction from $d$ to $d'$,
then $\calC$ will never be processed from $d$ to $d'$ again in the entire
algorithm (although $\calC$ may be processed later in the other direction from $d'$ to $d$).
\end{enumerate}

Initially $H=\emptyset$. Recall that $dis(d_s)=1$ and $B(d)=\{d_s\}$.
We consider the diagonal $d_s$ through $s$ as a degenerate junction cell. Specifically, we consider $d_s$ as two duplicate diagonals with one generating a rightward beam and the other generating a leftward beam from the entire $d_s$. We insert these two diagonals into $H$.
Clearly all algorithm invariants hold.
In the following we will describe the details of our algorithm. To
avoid the tedious discussion, we
will not explicitly explain that the algorithm invariants are
maintained after each step, but rather discuss it briefly later after the algorithm
description.

As long as $H$ is not empty, we repeatedly do the following.

By an extract-min operation on $H$, we obtain the diagonal $d^*$ of $H$ with
the smallest global-key and remove it from $H$. We assume $dis(d^*)=2i+1$ for some integer $i$.
If we were running the DN algorithm, we are currently working on the
$i$-th phase. There are two cases depending on
whether $B(d^*)=\emptyset$.

\subsubsection{$B(d^*)\neq \emptyset$}


We first discuss our algorithm for the case where $B(d^*)\neq
\emptyset$. If we were running the DN algorithm, we would be in the ``middle''
of the $i$-th phase (because $B(d^*)\neq \emptyset$ implies that all locally-outmost diagonals
have already been processed). So we
should continue the left-sweep and right-sweep of the $i$-th phase.
The first question is where we should start the sweep.
Let $S$ be the set of all diagonals in $H$ that have the same global-keys as $d^*$ (according to our way of comparing global-keys, for each diagonal $d\in S$, it holds that
$dis(d)=dis(d^*)$ and $B(d)\neq \emptyset$).
By the first algorithm invariant, all diagonals in $S$ have been correctly
labeled. Due to our corridor structure, the sweeps of the $i$-th phase
``paused'' at the diagonals in $S$. To continue the $i$-th phase, we
``resume'' the sweeps from these diagonals. Unlike the DN algorithm
where we complete the right-sweep before we start the left-sweep,
here, before the pause, we may have already done some left-sweep and
right-sweep. Hence, the two sweeps may be somehow ``interleaved'' and
our algorithm will need to take care of this situation.

Due to our way of breaking ties in $H$, the diagonals of $S$ can be found by
continuing the extract-min operations on $H$ in $O(|S|\log|H|)$ time (i.e., all diagonals of $S$ are removed from $H$). We also let $S$ contain $d^*$.
Let $S_R$ (resp., $S_L$) be the subset of the diagonals of $S$ whose beams are rightward (resp., leftward).
Intuitively, the right-sweep (resp., left-sweep) paused at the diagonals in $S_R$ (resp., $S_L$), and thus,
we resume it from the diagonals in $S_R$ (resp.,
$S_L$). Below we focus on the right-sweep, and the left-sweep is very similar.

Recall that in the right-sweep of the DN algorithm we use a heap
$H_R$ to guide the procedure. Here we do the same thing and
process the diagonals of $H_R$ from left to right.
We build a heap $H_R$ by inserting the diagonals of $S_R$, and the
``keys'' of diagonals in $H_R$ are
their $x$-coordinates such that the leftmost diagonal is at the root.
(Similarly, we build a heap $H_L$ on $S_L$ for the left-sweep.)

The algorithm essentially performs the $i$-th phase as the DN
algorithm. But since here the right-sweep and left-sweep may be
interleaved, in the following discussion, some diagonals may have two
sets of beams with opposite directions. However, our algorithm makes
sure that if a diagonal $d$ has
two sets of beams with opposite directions, it will not be in
$H$ (i.e., it has been removed from $H$), but in both $H_R$ and $H_L$
if $d$ has not been processed yet.
To differentiate the two sets of beams, we use $B_r(d)$ (resp., $B_l(d)$) to denote the beam set of any diagonal $d$ in $H_R$ (resp., $H_L$), meaning that the direction of
the beams is rightward (resp., leftward).

During the right-sweep, if we find a new diagonal $d$ that has the same global-key as $d^*$, then $d$ will be inserted to $H_R$ and $d$ will be removed from $H$ if it is already in $H$. Hence, all diagonals of $H_R$ have the same global-key as $d^*$. In fact, $H_R$ maintains all diagonals in junction cells that will be processed in the subsequent right-sweep of the $i$-th phase.

As long as $H_R$ is not empty, we repeatedly do the following.

We obtain the leftmost diagonal $d$ of $H_R$ (which is at the root of $H_R$)
and remove it from $H_R$. The beams of $B_r(d)$ may
enter either a junction cell or a corridor. In the sequel,
we discuss how to process $d$ in these two cases.

\paragraph{The beams of $B_r(d)$ entering a junction cell.}
Let $C$ denote the junction cell that the beams of $B_r(d)$ enter.
In this case, our way of processing $d$ is similar to
the DN algorithm. Let $e_l$ and $e_r$ denote the left and
right sides of $C$, respectively. Note that $d$ is on the left side
$e_l$. Recall that each cell side may have two diagonals. To process
$d$, we update the labels of all other diagonals of $C$, as
follows.

Suppose there is another diagonal $\hatd$ on $e_l$ (e.g., see Fig.~\ref{fig:cell}).

If $dis(\hatd)=\infty$, we set $dis(\hatd)=dis(d)$ and
$B(\hatd)=\emptyset$ with the leftward direction. Then, we insert $\hatd$ into
$H$.

If $dis(\hatd)\neq \infty$ but $dis(\hatd)>dis(d)$, then we set
$dis(\hatd)=dis(d)$ and $B(\hatd)=\emptyset$ (with the leftward direction).
Since $dis(\hatd)$ was greater than $dis(d)$ but not $\infty$, by the second algorithm invariant,
$\hatd$ is already in the heap $H$.
Hence, after $dis(\hatd)$ and $B(\hatd)$  are reset as above, we do
a ``decrease-key'' operation on $\hatd$ in $H$ since the global-key of $\hatd$ has been decreased.

If $dis(\hatd)\leq dis(d)$, we do nothing.

Next, we consider the diagonals on the right side $e_r$ of $C$. Depending on the
distance value $dis(e_r)$, there are several cases. Note that if
$dis(e_r)\neq \infty$, then
$e_r$ automatically got the value $dis(e_r)$ because a diagonal $d'$ on
$e_r$ was labeled with the same distance value.

\begin{enumerate}
\item
If $dis(e_r)=\infty$, then we set $dis(e_r)=2i+1$. If $e_r$ does not
have any diagonals, then we are done with processing $d$. Otherwise,
for each diagonal $d'$ on $e_r$, we set $dis(d')=2i+1$ and determine
$B_r(d)\cap d'$, i.e., the portion of $B_r(d)$ that can illuminate $d'$.

If $B_r(d)\cap d'\neq \emptyset$, then we set $B_r(d')=B_r(d)\cap
d'$ and insert $d'$ to $H_R$ (note that $d'$ has the same global-key as $d^*$).
Otherwise, we set $B(d')=\emptyset$ with the rightward direction and insert $d'$ into $H$, because the global-key of $d'$ is larger than that of $d^*$ and the right-sweep should stop at $d'$. Note that before the insertion, $d'$ was not in $H$ as $dis(d')$ was $\infty$.

\item
If $dis(e_r)\neq \infty$ but $dis(e_r)>2i+1$, then the algorithm is
similar as above. For each diagonal $d'$ on $e_r$, $dis(d')$ was equal to $dis(e_r)$, and now we set $dis(d')=2i+1$ and determine $B_r(d)\cap d'$.

If $B_r(d)\cap d'\neq \emptyset$,
we set $B_r(d')=B_r(d)\cap d'$ and insert $d'$ into $H_R$.
Since $dis(d')$ was equal to $dis(e_r)$, the global-key of $d'$ was larger than that of $d^*$ and $dis(d')$ was not $\infty$. By the second algorithm invariant, $d'$ was already in $H$.
Hence, we remove $d'$ from $H$. Note that we can perform a remove operation on $d'$ in
$H$ by doing a decrease-key operation followed by an extract-min
operation \cite{ref:CLRS09}.

If $B_r(d)\cap d'= \emptyset$, we set $B(d')=\emptyset$.
Since  $d'$ was already in $H$,
after setting $dis(d')=2i+1$ and  $B(d')=\emptyset$,  we need to do a
decrease-key operation on $d'$ in $H$.

\item
If $dis(e_r)<2i+1$, we do nothing.

\item
If $dis(e_r)=2i+1$, due to that the left and right sweeps are
interleaved, $e_r$ may have got labeled from the right-sweep, the
left-sweep, or both.
We discuss the three cases below. The algorithm for this case is also simple, but we need a few more words to explain why the algorithm works in that way.

\begin{enumerate}
\item
If $e_r$ got the value $dis(e_r)$ by the right-sweep only, as in
our discussions in the right-sweep of the DN algorithm,
this case happens because $e_r$ was illuminated by beams from
another diagonal $\hatd$ on the left side $e_l$. Hence, for each
diagonal $d'$ on $e_r$, $dis(d')$ and $B_r(d')$ have been set.
Further, if $B_r(d')=\emptyset$, then $d'$ is in $H$; otherwise $d'$
is in $H_R$.

For each diagonal $d'$ on $e_r$,
we first determine $B_r(d)\cap d'$, and then set $B_r(d')$ by merging the
original $B_r(d')$ with $B_r(d)\cap d'$. If $B_r(d')$ was non-empty
before the merge, we do nothing since it is already in $H_R$.
If $B_r(d')$ was empty before the
merge and becomes non-empty after the merge, then we insert $d'$ into
$H_R$ and remove it from $H$.
If $B_r(d')$ is still empty after the merge, then we do nothing
since it is already in $H$.

\item
If $e_r$ got the value $dis(e_r)$ by the left-sweep only, then $e_r$
was illuminated by beams from its right. As discussed in the DN
algorithm, if the right-sweep procedure sweeps a diagonal $d'$ on $e_r$,
the fact that $d'$ was swept already by the left-sweep procedure
should be ignored in the sense that we should keep propagating the
beams of $B_r(d)$ to the right of $d'$. The details are given below.

For each diagonal $d'$ on $e_r$, we first determine $B_r(d)\cap d'$.

\begin{itemize}
\item
If $B_r(d)\cap d'\neq\emptyset$, then we set $B_r(d')=B_r(d)\cap d'$
and insert $d'$ into $H_R$. If $d'$ was already in $H$, then we remove
it from $H$.

\item
If $B_r(d)\cap d'=\emptyset$, no beam from $B_r(d)$ illuminates $d'$.
At first sight, it seems that we should insert $d'$ into $H$
with an empty beam set and the rightward direction.
Below we elaborate on whether we should do so.

Note that $d'$ is a door of a corridor $\calC'$ that is locally on the
right of $d'$.  It is possible that $e_r$ got labeled
because the left-sweep was from $\calC'$, in which case
$B_l(d')$ is not empty by Observation \ref{obser:20},
and thus, $d'$ must be already in $H_L$ (since $dis(d')=2i+1$)
and we do not need to insert $d'$ to $H$ because $d'$ will be processed in the left-sweep of the current phase.

On the other hand, if $e_r$ has another diagonal $d''$, then it is
possible that $e_r$ got labeled because of the processing of the
corridor on the right of $d''$ during the left-sweep. In this case, $d'$ may have got
labeled because of the processing of $d''$, in which case as in
the DN algorithm the beam set of $d'$ must be empty and rightward, and thus $d'$
is already in $H$ with $B(d')=\emptyset$ and the rightward direction and we do
not insert $d'$ into $H$ again. But if $d'$ has not been labeled yet,
then we insert $d'$ into $H$ with $B(d')=\emptyset$ and the rightward
direction. Therefore, for the case where $e_r$ has another diagonal
$d''$, if $d'$ is already in $H$, we do
nothing; otherwise we insert $d'$ to $H$.

It is also possible that $e_r$ got labeled ``simultaneously'' because
of the processing of the two corridors on the right of $d'$ and $d''$,
in which case by Observation \ref{obser:20} we again have $B_l(d')\neq
\emptyset$, and thus $d'$ is already in $H_L$. Hence, we do not need to insert $d'$ to $H$.

In summary, for the case $B_r(d)\cap d'=\emptyset$,  if $d'$ is in
neither $H_L$ nor $H$, then we insert $d'$ into $H$ with
$B(d')=\emptyset$ and the rightward direction; otherwise we do
nothing.

\end{itemize}

\item
If $e_r$ got the value $dis(e_r)$ by both the right-sweep and the
left-sweep, this is a combination case of the above two cases.

For each diagonal $d'$ on $e_r$, we determine $B_r(d)\cap d'$, and
then set $B_r(d')$ by merging the original $B_r(d')$ with $B_r(d)\cap
d'$.

If $B_r(d')$ was non-empty before the merge, then we do nothing since
$d'$ is already in $H_R$.

If $B_r(d')$ was empty before the merge and becomes non-empty
after the merge, then we insert $d'$ into $H_R$. Further, if $d'$ is
in $H$, then we remove it from $H$.

If $B_r(d')$ is still empty after the merge, as in the above second
case, we do the following. If $d'$ is in neither
$H_L$ nor $H$, then we insert $d'$ into $H$ with
$B(d')=\emptyset$ and the rightward direction; otherwise we do
nothing.
\end{enumerate}
\end{enumerate}

We are done with processing $d$ when the beams $B_r(d)$ of $d$ enter
a junction cell.

\paragraph{The beams of $B_r(d)$ entering a corridor.}
Let $\calC$ denote the corridor that the beams enter.
We process $\calC$ using the beams of $B_r(d)$.
Again, one may assume we still use the DN algorithm to process $\calC$,
and later we will replace it by our corridor-processing
algorithm in Section \ref{sec:algocorridor}.

Let $\delta$ be the distance value labeled on the other door $d'$ of $\calC$
by the above processing and let $B'$ denote the corresponding beam set
on $d'$. Let $dis(d')$ and $B(d')$ be the original distance
value and beam set at $d'$ before the processing of $\calC$.
By the third algorithm invariant, this is the first time $\calC$ is processed in the direction from $d$ to $d'$. Hence, if $dis(d')\neq\infty$, then the value $dis(d')$ must be obtained by the sweep
from outside $\calC$, i.e., beams in $B(d')$ are towards the inside of $\calC$.

Due to the above processing of $\calC$, we have obtained another distance value $\delta$ and beam set $B'$ for $d'$. Hence, we need to update the label of $d'$ and possibly insert $d'$ to some heap.
Depending on the value of $dis(d')$, there are several cases.

\begin{enumerate}
\item
If $dis(d')$ is $\infty$, then we set $dis(d')=\delta$ and $B(d')=B'$.
If $\delta>2i+1$, then we insert $d'$ into $H$.
If $\delta=2i+1$, since $dis(d)=2i+1$, $d'$ must be
illuminated directly by the beams in $B_r(d)$ and the beams of $B(d')$ are still towards right.
By Observation \ref{obser:20}, $B'\neq\emptyset$.
Hence we obtain $B_r(d')=B'\neq\emptyset$ (we set $B_r(d')$ to $B'$
because the beams of $B'$ are rightward).  Finally, we insert $d'$ into $H_R$.

\item
If $dis(d')<2i+1$, then the global-key of $d'$ is smaller than that of $d^*$ because $dis(d^*)=2i+1$. By the first algorithm invariant, $d'$ has been
correctly labeled. Recall that the direction of $B(d')$ is towards the inside of $\calC$.
Also by the first algorithm invariant, $d$ is
correctly labeled, and the direction of the beam set of $B_r(d)$ is
towards the inside of $\calC$. This means we have computed complete information on the
two doors of $\calC$ for the min-link v-v-paths from $s$ to the points
inside $\calC$. Hence, we can do a ``post-processing'' step to
compute the v-v-map in $\calC$ by using the beams of $B_r(d)$ and
$B(d')$. We will give a {\em corridor-post-processing} algorithm for this step later in Section \ref{sec:algocorridor}.

\item
If $dis(d')=2i+1$, then $d'$ has been labeled in the current phase.

If $B(d')\neq \emptyset$, then the global-key of $d'$ is the same as that of $d^*$. By
the first algorithm invariant, $d'$ has been correctly labeled. Then, as
above, since both $d$ and $d'$ have been correctly
labeled, we do a ``post-processing'' to compute the
v-v-map in $\calC$ using $B_r(d)$ and $B(d')$.

If $B(d')= \emptyset$, then the global-key of $d'$ is strictly larger than that of
$d^*$. By the second algorithm invariant, $d'$ is already in $H$.
If $\delta>2i+1$, then we do nothing.
If $\delta=2i+1$, then as in the above first case, we set $B_r(d')=B'\neq\emptyset$;
finally, we insert $d'$ into $H_R$ and remove $d'$ from $H$.

\item
The remaining case is when $dis(d')\neq \infty$ and $dis(d')>2i+1$.
We claim that this case can never happen.

Indeed, assume to the contrary that this case happens. Recall that the beams of $B(d')$ are towards the inside of the corridor $\calC$. Let $C$ be the junction cell that contains $d'$. Without loss of generality, assume $d'$ is on the left side of $C$.
Clearly, $d'$ got labeled after some diagonal $d''$ in $C$ was processed.
Since the beams of $d''$ that illuminate $d'$ must be towards the cell $C$, they
are from the corridor that is bounded by $d''$. By Observation
\ref{obser:20}, the beam set of $d''$ is not empty. Hence, it must be
the case that $dis(d'')=dis(d')>2i+1$. Since we use the heap $H$ to guide the main algorithm and we are currently processing the diagonal $d$ with $dis(d)=2i+1$, all diagonals of junction cells that have been processed must have distance values at most $2i+1$. However, the above shows that the diagonal $d''$
has been processed with $dis(d'')>2i+1$, incurring contradiction. 
Therefore, the case where $dis(d')\neq \infty$ and
$dis(d')>2i+1$ cannot happen.
\end{enumerate}

The above finishes the processing of the diagonal $d$. The
right-sweep procedure is done after the heap $H_R$ becomes empty.
Afterwards we do the left-sweep from the diagonals of $S_L$ using the heap $H_L$ in the
symmetric way. We omit the details.

This finishes our algorithm in the case in which $B(d^*)$ is not empty
for the diagonal $d^*$ at the root of $H$. We briefly discuss that
after the above processing our algorithm invariants still hold.

Indeed, before the diagonals of $S_R$ are processed, if we were running
the DN algorithm, at some moment during the $i$-th phase, the
diagonals of $S_R$ would be labeled the same as in our algorithm. Our algorithm
processes the diagonals of $S_R$ in a way consistent with the DN
algorithm. Based on our previous discussion on how the sweep procedures of the DN
algorithm can sweep the corridors, after all diagonals of $S_R$ are
processed, the diagonals in $H$ with the smallest global-key must have been
correctly labeled because if we were running the DN algorithm, at some
moment the algorithm would label them in the same way. Also, since we
process diagonals in the order of their global-keys, any
diagonal that has global-key smaller than that of the root of $H$ must have
been processed and labeled correctly. Hence, the first algorithm
invariant follows. For the second invariant, whenever a diagonal has
its distance value set to non-infinity for the first time, it is
always inserted into $H$.
For the third invariant, as discussed in
the algorithm description, a corridor $\calC$ is processed only if a
door $d$ of $\calC$ is processed and beams of $B(d)$ are towards
$\calC$. Although $d$ may be processed in the algorithm twice, it is
processed only once in either the right-sweep or the left-sweep.
Hence, with beams towards the inside of $\calC$, $d$ is only
processed once in the entire algorithm, implying that $\calC$ is
processed only once in the entire algorithm in the direction from $d$ to
the other door of $\calC$.

\subsubsection{$B(d^*)=\emptyset$}

In the sequel, we discuss the case where $B(d^*)=\emptyset$.
The algorithm is simpler in this case.  First,
we find the set $S$ of diagonals in $H$ that have the same global-key as $d^*$, which can be done in $O(|S|\log |H|)$ time by keeping doing the extract-min operations (i.e., all diagonals of $S$ are removed from $H$). We also let $S$ contain $d^*$. According to our way of comparing global-keys, all
diagonals of $S$ have empty beam sets. If we were running the DN
algorithm, the diagonals of $S$ would be locally-outmost
and we would be about to start the $(i+1)$-th phase (not the $i$-th phase).
As in the previous case, we run the two sweep procedures
starting from the diagonals of $S$.
Let $S_R$ (resp., $S_L$) be the subset of diagonals of $S$ whose beam directions
are rightward (resp., leftward).  We build a min-heap $H_R$ (resp., $H_L$)
on the diagonals of $S_R$ (resp., $S_L$).
Below, we only discuss the right-sweep since the left-sweep is similar.

Since we are doing the right-sweep in the $(i+1)$-th phase, each
diagonal of $S_R$ will generate a beam from the entire diagonal, and
all new diagonals illuminated in the right-sweep will
get distance value $2i+3$ instead of $2i+1$.
From now on, we associate each diagonal of $S_R$ with the beam, i.e.,
for each $d\in S_R$, $B_r(d)$ now consists of the only beam on $d$ although it was
empty in the $i$-th phase. Since each diagonal of $S_R$ originally got an empty beam set (at the end of the $i$-th phase), by Observation \ref{obser:20}, the beam of $B_r(d)$ cannot be towards a junction cell and thus it must be towards the inside of a corridor.

As long as $H_R$ is not empty, we repeatedly do the following.

We obtain the leftmost diagonal $d$ of $H_R$ (which is at the root) and remove it from $H_R$. Let $\calC$ denote the corridor that the beams of $B_r(d)$ enter.
We process the corridor $\calC$ using the beams of $B_r(d)$.
Again, one may assume we still use the DN algorithm to process $\calC$,
and later we will replace it by our corridor-processing
algorithm in Section \ref{sec:algocorridor}.

Let $\delta$ be the distance on the other door $d'$ obtained by the above
processing and let $B'$ be the corresponding beam set.
Let $dis(d')$ and $B(d')$ be the original distance value and beam set
at $d'$.
Again, by the third algorithm invariant, this is the the first
time $\calC$ is processed in the direction from $d$ to $d'$; hence,
if $dis(d')\neq \infty$, then $d'$ must be labeled by a sweep from outside $\calC$
and the beams of $B(d')$ must enter $\calC$.
Depending on the value of $dis(d')$, we may need to update the label of $d'$ in several cases.

\begin{enumerate}
\item
If $dis(d')=\infty$, we set $dis(d')=\delta$ and $B(d')=B'$. Note that
$\delta\geq 2i+3$. Hence, the global-key of $d'$ is strictly larger than that of $d^*$, which has distance value $2i+1$.
We insert $d'$ into $H$ (not $H_R$).


\item
If $dis(d')\leq 2i+1$, then since $B(d^*)=\emptyset$, the global-key of $d'$ is no larger than that of $d^*$ regardless of whether $B(d')$ is empty or not. By the first algorithm invariant, $d'$ has been correctly labeled. We do a ``post-processing'' to compute the
v-v-map in $\calC$ by using the beams of $B_r(d)$ and $B(d')$. Again, we will give a {\em corridor-post-processing} algorithm for this step later in Section \ref{sec:algocorridor}.

%
%
%
%

\item
The remaining case is when $dis(d')\neq \infty$ and $dis(d')>2i+1$.
By the same argument as in the previous case where
$B(d^*)\neq\emptyset$, this case cannot happen.
\end{enumerate}

The above describes the right-sweep procedure. The
left-sweep is similar.

This finishes our discussion in the case in which $B(d^*)$ is empty for
the diagonal $d^*$ at the root of $H$. As in the previous case, all
algorithm invariants hold.


The algorithm finishes if all three heaps $H$, $H_L$, and $H_R$ become
empty. After that, for each diagonal $d$ in a junction cell, $dis(d)$
and $B(d)$ have been correctly computed. During the algorithm
some corridors have been labeled correctly while others are left for
post-processing.

Specifically, consider any corridor $\calC$ and let $d_1$ and $d_2$ be
its two doors with their beam sets $B(d_1)$ and $B(d_2)$. If $\calC$ is not left for a post-processing, then $\calC$ has been processed either from $d_1$ to $d_2$ or from $d_2$ to $d_1$ and the v-v-map in $\calC$ has been computed after the processing. Suppose the above processing is from $d_1$ and $d_2$. Then, $\calC$ is processed using the beams of $B(d_1)$ and $B(d_2)$ is obtained after the processing.
We will show in Section
\ref{sec:algocorridor} that our corridor-processing algorithm on $\calC$ runs in $O(m+(h_1-h_2+1)\log
h_1)$ time, where $m$ is the number of vertices of $\calC$,
$h_1=|B(d_1)|$, and $h_2=|B(d_2)|$.
If $\calC$ is left for a post-processing, i.e., to compute the
v-v-map in $\calC$ by using $B(d_1)$ and $B(d_2)$, we will show in
Section \ref{sec:algocorridor} that our
corridor-post-processing algorithm runs in $O(m+h_1\log h_1+h_2\log
h_2)$ time.


\begin{lemma}
Given $\vtd(\calF)$, our algorithm computes the v-v-map on $\vtd(\calF)$ in $O(n+h\log h)$ time.
\end{lemma}
\begin{proof}
Recall that the beams in our algorithm are generated by
locally-outmost diagonals. We say that two beams are {\em different}
if they do not have the same generator. We say a diagonal is {\em generated}
by a corridor if the generator of the
diagonal is in the corridor.

We first prove a {\em claim} that the
number of different beams at the diagonals in all junction cells is $O(h)$. To see this, a key
observation is that since a corridor is a simple rectilinear polygon,
it can only generate at most two new beams that can
go out of the corridor in the entire algorithm. Specifically, consider a
corridor $\calC$ with doors $d_1$ and $d_2$. Suppose the algorithm
processes $\calC$ in the direction from $d_1$ to $d_2$. Then, as will be seen later in Section \ref{sec:algocorridor}, if some beams of $B(d_1)$ can directly illuminate $d_2$, then all beams of
$B(d_2)$ are from $B(d_1)$ (although some beams may become narrowed) and there is no new beam generated by $\calC$;
otherwise all beams of $B(d_1)$ terminate inside $\calC$ and
$B(d_2)$ will only have one beam coming out of the corridor through $d_2$ (this is because $\calC$ is a simple rectilinear polygon).
Since any corridor can be processed at most twice,
it can generate at most two new beams.
Since there are $O(h)$ corridors, the total number of beams on
the diagonals of junction cells is $O(h)$.

To obtain the running time of the algorithm, we analyze the
time we spent on junction cells and the corridors separately. Due to the above
claim, the size of $B(d)$ for each diagonal $d$ in any junction cell
is $O(h)$. Since there are $O(h)$ diagonals in all junction cells,
the total time we spent on processing them is $O(h\log h)$.

For the time we spent on all corridors, it is the sum of the time
of our corridor-processing algorithm and corridor-post-processing
algorithm on all corridors.

Let $S$ denote the set of the diagonals in all junction cells.
Define $\calC$, $d_1$, $d_2$, $B(d_1)$, $B(d_2)$, $m$, $h_1$, and $h_2$ the same as before. Suppose
$\calC$ has been post-processed. Then the
corridor-post-processing algorithm on $\calC$ takes
$O(m+h_1\log h_1+h_2\log h_2)$ time. The sum of the term $m$ overall all
corridors is $O(n)$. All beams of $B(d_1)$ and $B(d_2)$ are terminated
inside $\calC$ after the post-processing. Since the number of beams that are terminated inside corridors is no more than the number of different beams at the diagonals in all junction cells, by the above claim, the number of beams on the
diagonals of $S$ that are terminated inside corridors is $O(h)$. Hence, the sum of
$h_1$ (resp., $h_2$) over all corridors that have been post-processed
is $O(h)$, and the sum of $h_1\log h_1+h_2\log h_2$ over all
corridors that have been post-processed is $O(h\log h)$. Note that
each corridor can be post-processed at most once. This proves
that the total time of the corridor-post-processing algorithm in the
entire algorithm is $O(n+h\log h)$.

We can use the similar approach to analyze the total time of the
corridor-processing algorithm. Suppose $\calC$ has been processed by
the corridor-processing algorithm in
the direction from $d_1$ to $d_2$. Then, the running time of the
algorithm on $\calC$ is $O(m+(h_1-h_2+1)\log h_1)$ time. Similarly,
the sum of $m$ over all corridors is $O(n)$.
After the processing of $\calC$, the number of beams of $B(d_1)$ that
have been terminated in $\calC$ is at least $h_1-h_2$. Since the the number of beams on the
diagonals of $S$ that are terminated inside corridors is $O(h)$, the sum of the term
$h_1-h_2$ over all corridors that have been processed is $O(h)$. The
sum of the additional term $1$ over all corridors is clearly $O(h)$.
Therefore, although a corridor may be processed twice, the total time
of the corridor-processing algorithm in the entire algorithm is
$O(n+h\log h)$.

The lemma thus follows.\qed
\end{proof}

The above computes the v-v-map on $\vtd(\calF)$. Again, the above algorithm only labels diagonals.
Using the similar approaches as discussed in Section \ref{sec:labelcell},
we can also label cells and maintain path information within the same running time asymptotically. We omit the details.

The other three maps can be computed similarly. For computing the
h-v-map on $\vtd(\calF)$, one difference is on the initial steps, as follows. Initially, we let $s$ generate two beams that are two horizontal rays towards right and left, respectively. We set the distance value of $d_s$ to $0$, where $d_s$ is the vertical diagonal through $s$. Then, we consider $d_s$ as two duplicate diagonals associated with the above two beams respectively, and insert the two duplicate diagonals into the heap $H$. The remaining algorithm is the same as before except that we replace the distance values $2i+1$ and $2i+3$ in the algorithm description with $2i$ and $2i+2$, respectively. The h-h-map and v-h-map on $\htd(\calF)$ can be computed in symmetric ways.

We thus obtain the main result of this paper.

\begin{theorem}
Given a set of $h$ pairwise-disjoint rectilinear polygonal obstacles with a total of $n$ vertices in the plane, after the free space is triangulated,
we can construct a link distance map with respect to any given source point $s$ in $O(n+h\log h)$ time, such that
for any query point $t$, the link distance of a min-link rectilinear \st\ path $\pi$ can be computed in $O(\log n)$ time
and $\pi$ can be output in additional time linear in its link distance.
\end{theorem}

\section{The Algorithms for Processing Corridors}
\label{sec:algocorridor}

In this section, we present our corridor-processing algorithm and
corridor-post-processing algorithm. In fact, as shown later, the second algorithm is
simply to run the first one twice from the two doors. Below, we discuss the
corridor-processing algorithm first.

%
%

Let $\calC$ be any corridor with two doors $d_1$ and $d_2$ and
their beam sets $B(d_1)$ and $B(d_2)$. Let $m$ be the number of vertices of $\calC$,
$h_1=|B(d_1)|$, and $h_2=|B(d_2)|$.
Suppose we
want to process $\calC$  to compute the labels of all diagonals,
by using the beams in $B(d_1)$ and the distance
value $dis(d_1)$ of $d_1$. In particular,
we want to obtain the beam set $B(d_2)$ and the distance value
$dis(d_2)$ for $d_2$.
Our algorithm is conceptually similar to the DN algorithm if we apply
it on $\calC$, but our approach is more efficient due to a better implementation by making
use of the simplicity of the corridor (one crucial property is that there will be no merge
operations on the beam sets).


Denote by $\vtd(\calC)$ the vertical decomposition of $\calC$ (e.g.,
see Fig.~\ref{fig:simDec}). Note that although $d_1$ and $d_2$ are
now on the boundary of $\calC$, we still consider them as diagonals.
Let $C(d_1)$ denote the cell of $\vtd(\calC)$ that contains $d_1$.
Starting from $C(d_1)$, we
propagate the beams from $B(d_1)$ to all other cells one by one until
all diagonals have been labeled.


\begin{figure}[t]
\begin{minipage}[t]{\linewidth}
\begin{center}
\includegraphics[totalheight=1.2in]{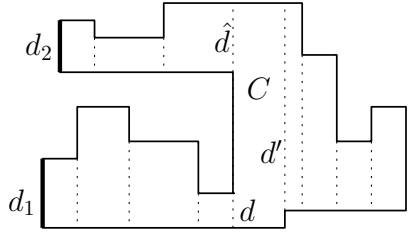}
\caption{\footnotesize Illustrating $\vtd(\calC)$: the two doors are
shown with thick segments.  }
\label{fig:simDec}
\end{center}
\end{minipage}
\vspace{-0.15in}
\end{figure}

Consider any diagonal $d$ of $\vtd(\calC)$ such that $d$ is not $d_1$
or $d_2$. Since $\calC$ is simply connected,
$d$ divides $\calC$ into two polygons, and we use $\calC(d)$ to refer
to the one that does not contain $d_1$. Clearly, for any point $t\in
\calC(d)$, any path from $d_1$ to $t$ must intersect $d$. Consider any
cell $C$ of $\vtd(\calC)$. The cell $C$ may have diagonals on both its
two vertical sides. It is not difficult to see that $C$ has one and only one diagonal $d$
such that $C$ is in $\calC(d)$, and we call that diagonal the {\em
entrance diagonal} of $C$. Other diagonals of $C$ are called {\em
exit diagonals} of $C$.  For example, in Fig.~\ref{fig:simDec}, $d$ is the
entrance diagonal of $C$, and $\hatd$ and $d'$ are exit
diagonals. Note that every diagonal
is an entrance diagonal of one and only one cell. In particular,
we consider $d_1$ as the entrance diagonal of $C(d_1)$.

A general step of our algorithm works as follows. Consider a cell $C$
and suppose the entrance diagonal $d$ of $C$ has been labeled, i.e.,
$dis(d)$ and $B(d)$ are available (and the beams of $B(d)$ are stored
in a balanced binary search tree $T(d)$).
Initially, $C$ is the cell $C(d_1)$ and $d$ is $d_1$.
Our goal for processing $d$ is to label
all exit diagonals in $C$.

Denote by $e_l$ and $e_r$ the
left and right sides of $C$, respectively. Without loss of generality, we assume $d$
is on $e_l$ (e.g., see Fig.~\ref{fig:simDec}).
Thus, the beams of $B(d)$ are rightward.
The distance values of the exit diagonals are easy to compute. If $B(d)$
is empty, then we need to generate a single beam from the entire $d$ and
every exit diagonal of $C$ obtains the distance value $dis(d)+2$; otherwise,
every exit diagonal of $C$ obtains the distance value $dis(d)$.
Below, we focus on computing the beam sets of the exit
diagonals.

Suppose $e_l$ has another diagonal $\hatd$. Then, regardless of
whether $B(d)$ is empty, we set $B(d)=\emptyset$ (and $T(d)=\emptyset$), which
can be done in only constant time.

Now consider any diagonal $d'$ on $e_r$.  First, we want to compute $B(d)\cap
d'$, i.e., the portions of the beams of $B(d)$ that can illuminate $d'$, and set $B(d')=B(d)\cap d'$. If $B(d)=\emptyset$, then since we generate a single beam from
$d$, $B(d')=B(d)\cap d'$ has at most one beam
and can be computed in constant time.

In the sequel we consider the case where
$B(d)\neq \emptyset$. Our main effort is on handling this case.
Recall that the beams of $B(d)$ are stored in a
balanced binary search tree $T(d)$. We add two special pointers to
$T(d)$ that point to
the lowest beam and the highest beam of $B(d)$ respectively so that we
can access these beams in constant time. Note that with these special
pointers, we can still perform the previous operations on
$T(d)$ each in logarithmic time.

Depending on whether $e_r$ has one or two diagonals, there are two
cases.

If $e_r$ has only one diagonal $d'$, let $e$ be the obstacle edge on $e_r$ such that $d'$
is the vertical extension of $e$ (e.g., see Fig.~\ref{fig:case10}). Clearly, $e_r$ is the union of $d'$ and $e$.
Without loss of generality, assume $e$ is lower than $d'$.

Suppose $b$ is a beam in $B(d)$ and $b'$ is the rightward
projection of $b$ on $e_r$. For any line segment $l$ on $e_r$,
we say that $b$ {\em intersects} $l$ if $b'$
intersects $l$ properly, and $b$ {\em fully intersects} $l$ if $b'$ is
contained in $l$.

Below we let $b$ denote the lowest beam of
$B(d)$, which can be obtained in constant time by using the special
pointers on $T(d)$.
We first check whether $b$ intersects $e$.
Depending on how $b$ intersects $e$, there are three cases. The correctness of our setting in all these cases is based on that $e_r$ consists of $d'$ and $e$ from top to bottom (e.g., see Fig.~\ref{fig:case10}).

\begin{enumerate}
\item
If $b$ does not intersect $e$, then every beam of $B(d)$ fully
intersects $d'$. Thus we have $B(d')=B(d)$ and $T(d')=T(d)$. This can be done in constant time.

\item
If $b$ intersects $e$ but does not fully intersects $e$, then we
set $B(d')=B(d)$ and $T(d')=T(d)$, but we also change $b$'s length to that of its
portion intersecting $d'$.
This can be done in constant time.

\item
If $b$ fully intersects $e$ (so $b$ ``terminates'' at $e$), then depending on whether $B(d)$ has only one beam, there are further two subcases.

If $B(d)$ has only one beam, which is $b$, then we simply set $B(d')=\emptyset$ and $T(d')=\emptyset$ since $b$ terminates at $e$.

Otherwise, we have $|B(d)|\geq 2$. In this subcase, as before we use a split
operation on $T(d)$ to obtain $T(d')$ and $B(d')=B(d)\cap d'$ in $O(\log |B(d)|)$ time.
Note that since $b$ terminates at $e$, $b$ is not in $B(d')$, and thus
$|B(d')|\leq |B(d)|-1$. This will be useful to our time analysis on the split operations.
\end{enumerate}

\begin{figure}[t]
\begin{minipage}[t]{0.49\linewidth}
\begin{center}
\includegraphics[totalheight=1.2in]{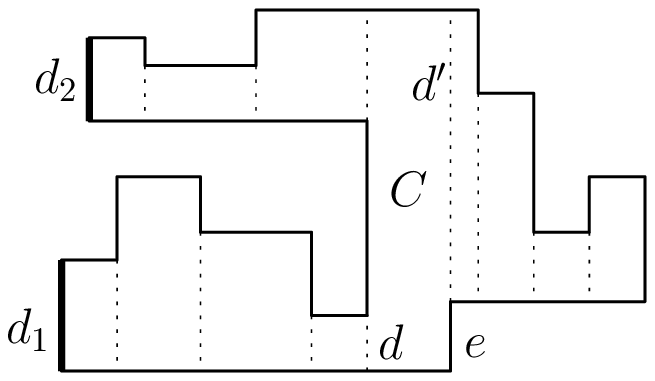}
\caption{\footnotesize The right side $e_r$ of $C$ consists of $d'$ and $e$ from top to bottom.}
\label{fig:case10}
\end{center}
\end{minipage}
\hspace*{0.02in}
\begin{minipage}[t]{0.49\linewidth}
\begin{center}
\includegraphics[totalheight=1.2in]{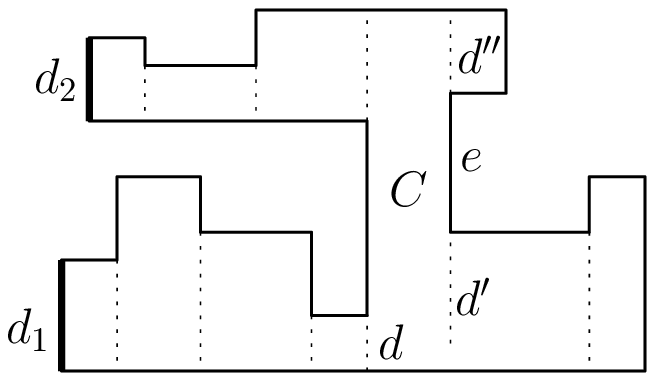}
\caption{\footnotesize The right side $e_r$ of $C$ consists of $d''$, $e$, and $d'$ from top to bottom.  }
\label{fig:case20}
\end{center}
\end{minipage}
\vspace{-0.15in}
\end{figure}

If $e_r$ has another diagonal $d''$,
without loss of generality, we assume $d'$ is the lower one.
Let $e$ be the obstacle edge on $e_r$. Then $e_r$  consists of $d''$,
$e$, and $d'$ from top to bottom (e.g., see Fig.~\ref{fig:case20}).

Again, let $b$ denote the lowest beam of $B(d)$. Depending on how $b$
intersects $d'$, there are three cases. In every case below, we will also obtain a tree $T'$
of beams that will be used later to compute the beams of $d''$.

\begin{enumerate}
\item
If $b$ does not intersect $d'$, then we have $B(d')=\emptyset$ and
$T(d)=\emptyset$. We set $T'=T(d)$.

\item
If $b$ intersects $d'$ but does not fully intersects $d'$, then $B(d')$ consists of a
single beam that is the portion of $b$ intersecting $d'$.
It is straightforward to construct $T(d')$.
We set $T'=T(d)$ but also change the length of $b$ to that of the portion of $b$
not intersecting $d'$. All above can be done in constant time.

\item
If $b$ fully intersects $d'$, then we further check whether the
{\em highest} beam $b^*$ of $B(d)$ intersects $d'$.

If $b^*$ fully intersects $d'$, then we have $B(d')=B(d)$ and
$T(d')=T(d)$. Also, $T'=\emptyset$.

If $b^*$ intersects $d'$ but does not fully intersect $d'$, then we
set $B(d')=B(d)$ and $T(d')=T(d)$ but also change $b^*$'s length to
that of its portion intersecting $d'$. Also, we let $T'$  only
include the portion of $b^*$ not intersecting $d'$. All this can be done in constant time.

If $b^*$ does not intersect $d'$,
then as before we use a split operation that split
$T(d)$ into two trees: $T(d')$, which consists of the beams of
$B(d')=B(d)\cap d'$, and $T'$,
which consists of the rest of the beams of $B(d)$. This can be done in $O(\log |B(d)|)$ time.
Let $B'$ be the set of beams in $T'$.
Note that since $b$ fully intersects $d'$, $b$
is not in $B'$, and since $b^*$ does not intersect $d'$, $b^*$ is
not in $B(d')$. Hence, we have $|B'|\leq |B(d)|-1$ and $|B(d')|\leq |B(d)|-1$, which will be useful to our time analysis on the split operations. Further, since $b$ fully intersects $d'$ but $b^*$ does not intersect $d'$, we have $b'\neq b^*$ in this case, implying that $|B(d)|\geq 2$.
\end{enumerate}

Next, we compute $B(d'')$ and $T(d'')$ from the beams in the tree $T'$, in
the same way as the previous case where $e_r$ only contains one
diagonal. Namely, we first check how the lowest beam of $T'$
intersects $e$ and then proceed accordingly for the three cases.

The above labels all exit diagonals of the cell $C$. Based on our above algorithm, an easy observation is that if $|B(d)|=1$, then $|B(d')|\leq 1$ for each exit diagonal $d'$ of $C$.

Next, from each exit diagonal of $C$,
we proceed with the same procedure. The algorithm is done once
all diagonals have been labeled. The following lemma analyzes the running time of the algorithm.

\begin{lemma}
Our corridor-processing algorithm on $\calC$ runs in $O(m+(h_1-h_2+1)\log h_1)$ time.
\end{lemma}
\begin{proof}
We make a few
observations on our algorithm. First,
for any diagonal $d$, if computing $B(d)$ does not need a split
operation, then it takes only constant time to do so.
Second, for any diagonal $d$ of $\calC$, $|B(d)|\leq |B(d_1)|=h_1$
always hold, and if $|B(d)|\geq 2$, then the beams of $B(d)$ are from
$B(d_1)$. Third, for any diagonal $d$, if $B(d)$ is obtained by a split
operation on a beam set $B$, then $|B|\geq 2$ and $|B(d)|\leq |B|-1$.

The first two observations imply that if there are $k$ split
operations in the entire algorithm, then the total running
time is $O(m+k\log h_1)$. We claim that $k=O(h_1-h_2+1)$.

Indeed, by the second and the third observation, if a split operation is performed on $B(d)$ for some diagonal $d$, then all beams of $B(d)$ must be originally from $B(d_1)$, and further, the split operation splits $B(d_1)$ into two sets such that the number of beams in each set is at most $|B(d)|-1$. Therefore, it can be verified that if $h_2\geq 2$, then $k=O(h_1-h_2)$, and otherwise, $k=O(h_1)$. In either case, $k=O(h_1-h_2+1)$.

We conclude that the algorithm runs in $O(m+(h_1-h_2+1)\log h_1)$ time.\qed
\end{proof}

Next, we present our corridor-post-processing algorithm. Again,
consider the corridor $\calC$ as above, but now $B(d_2)$ and
$dis(d_2)$ are also given as input and the beams of $B(d_2)$ are towards
the inside of $\calC$.
Our goal is to compute the distance values for
all diagonals of $\calC$ by using the beams of $B(d_1)$ and $B(d_2)$.
An easy solution is to use our above corridor-processing algorithm to process
$\calC$ twice, once only using $dis(d_1)$ and $B(d_1)$ and once only using
$dis(d_2)$ and $B(d_2)$. Then, each diagonal $d$ has
been labeled twice, and we finally set $dis(d)$ to the smaller distance value
labeled above. Clearly, the total running time is bounded by
$O(m+h_1\log h_1+h_2\log h_2)$.

\section{Conclusions}
\label{sec:conclude}

We presented a new algorithm for computing minimum
link paths in rectilinear polygonal domains. The algorithm
matches the $\Omega(n+h\log h)$ time lower bound \cite{ref:MaheshwariLi00}
if the triangulation of the free space can be done optimally in $O(n+h\log h)$ time.
Our algorithm can also build the minimum-link distance map for a fixed source point to answer minimum link
path queries.

It would be interesting to see if our approach can be extended to work in a $C$-oriented world \cite{ref:AdegeestMi94,ref:HershbergerCo94,ref:MitchellMi14}; one main stumbling block is that the DN algorithm does not generalize to $C>2$ orientations. Another intriguing open problem is whether a result similar to ours, could be possible for paths and domains with unrestricted orientations \cite{ref:MitchellMi92}, e.g., whether there exists an $O(g(n)+f(h))$-time algorithm for the problem where $g$ is subquadratic (the 3SUM-hardness reduction \cite{ref:MitchellMi14} for the problem shows that $f$ essentially cannot be subquadratic).

%



\bibliographystyle{plain}
\bibliography{reference}

%



\end{document}